\documentclass[sigconf, anonymous=false]{acmart}
\makeatletter
\@ACM@authorversiontrue
\makeatother

\AtBeginDocument{%
  \providecommand\BibTeX{{%
    \normalfont B\kern-0.5em{\scshape i\kern-0.25em b}\kern-0.8em\TeX}}}

\copyrightyear{2024}
\acmYear{2024}
\setcopyright{acmlicensed}\acmConference[BSCI:'24]{The 6th ACM International Symposium on Blockchain and Secure Critical Infrastructure}{July 2, 2024}{Singapore, Singapore}
\acmBooktitle{The 6th ACM International Symposium on Blockchain and Secure Critical Infrastructure (BSCI:'24), \linebreak July 2, 2024 Singapore, Singapore}
\acmDOI{10.1145/3659463.3660021}
\acmISBN{979-8-4007-0638-7/24/07}
\newcommand{\wg}[1]{#1}

\begin{document}

\title{A Game Theoretic Analysis of Validator Strategies in Ethereum 2.0}

\author{Chien-Chih Chen}
\email{j2255che@uwaterloo.ca}
\orcid{0000-0003-1170-1796}
\affiliation{%
  \department{Department of Electrical and Computer Engineering}
  \institution{University of Waterloo}
  \streetaddress{200 University Ave W}
  \city{Waterloo}
  \state{Ontario}
  \country{Canada}
  \postcode{N2L 3G1}
}

\author{Wojciech Golab}
\email{wgolab@uwaterloo.ca}
\orcid{0000-0002-8891-256X}
\affiliation{%
  \department{Department of Electrical and Computer Engineering}
  \institution{University of Waterloo}
  \streetaddress{200 University Ave W}
  \city{Waterloo}
  \state{Ontario}
  \country{Canada}
  \postcode{N2L 3G1}
}

\begin{abstract}
Ethereum 2.0 is the second-largest cryptocurrency by market capitalization and a widely used smart contract platform. Therefore, examining the reliability of Ethereum 2.0's incentive mechanism is crucial, particularly its effectiveness in encouraging validators to adhere to the Ethereum 2.0's protocol. This paper studies the incentive mechanism of Ethereum 2.0 and evaluates its robustness by analyzing the interaction between block proposers and attesters in a single slot. To this end, we use Bayesian games to model the strategies of block proposers and attesters and calculate their expected utilities. Our results demonstrate that the Ethereum 2.0 incentive mechanism is incentive-compatible and promotes cooperation among validators. We prove that a Bayesian Nash equilibrium and an ex ante dominant strategy exist between the block proposer and attesters in a single slot. Our research provides a solid foundation for further analysis of Ethereum 2.0's incentive mechanism and insights for individuals considering participation as a validator in Ethereum 2.0.
\end{abstract}




\settopmatter{printfolios=true}
\maketitle
\section{Introduction}
As computing tasks have become more complex and networks have become faster, distributed systems are widely adopted in building computer software and applications. The ownership of traditional distributed systems is centralized, which means certain central companies or authorities control them. This centralization is evident in how, although the physical servers are deployed in different regions, the data control belongs to a specific enterprise. Conventional centralized distributed systems have a central authority that is not always reliable, which promotes the development of the decentralized distributed system \cite{nakamoto2008bitcoin}. Decentralization means that a single enterprise no longer owns data; instead, data governance is delegated to each participating entity \cite{raval2016decentralized, cai2018decentralized}. In 2008, Nakamoto consensus was proposed in the Bitcoin white paper \cite{nakamoto2008bitcoin}, which is recognized as a significant achievement in realizing the zero-trust decentralized payment system. However, Bitcoin lacks the concept of smart contracts, which limits its application. Therefore, Ethereum \cite{vitalikannotatedspec} was launched in 2015, providing a Turing complete programming language called Solidity so that people can build various applications in Ethereum. 

Permissionless blockchains, also known as public blockchains, allow anyone to participate in the consensus decision process used to verify transactions \cite{ 8525392, helliar2020permissionless}. Bitcoin \cite{nakamoto2008bitcoin} and Ethereum \cite{vitalikannotatedspec} are examples of permissionless blockchains. In such blockchains, incentive mechanisms are crucial because they must motivate miners or validators to contribute their computing power to maintain the network and validate transactions. Hence, an appropriate incentive mechanism needs to be designed to encourage participants to act honestly and to deter the presence of Byzantine nodes. This area of research falls within the domain of mechanism design \cite{alma9938932363505162}.

One of the core elements of mechanism design is incentive compatibility, which ensures that agents participating in games adhere to established rules \cite{alma9938932363505162}. In Ethereum 2.0 (Eth2)\footnote{In this document, \emph{Ethereum 2.0} (also referred to as \emph{Eth2}) denotes the version of Ethereum following the transition from a Proof of Work (PoW) to a Proof of Stake (PoS) consensus mechanism. In 2022, the Ethereum Foundation renamed the original PoW-based Ethereum to the \emph{"execution layer"} and Eth2 to the \emph{"consensus layer"}. To avoid confusion with the term \emph{"consensus layer"}, this document continues to use \emph{Ethereum 2.0} or \emph{Eth2} to describe the post-transition Ethereum.
}, validators are required to keep their devices operational 24 hours a day, seven days a week, to ensure the operation of Eth2's consensus rules \cite{vitalikannotatedspec}. If maintaining connectivity to execute validators' tasks becomes the best response in a Bayesian Nash equilibrium (BNE) or a dominant strategy for validators, then Eth2's incentive mechanism can be considered to have incentive compatibility. As a result of this consideration, this study investigates the incentive mechanism for validators in Eth2 and analyzes its compatibility. We focus on the strategic interaction between block proposers and attesters within a single slot. The research is conducted in two steps: First, we model the interaction between proposers and attesters within a single slot as a Bayesian game and conclude that cooperation is not only a Bayesian Nash equilibrium (BNE) but also an ex ante dominant strategy for a validator. Next, based on the equilibrium analysis, we conclusively state that the current reward and penalty schemes in Eth2 encourage block proposers and attesters to maintain the connectivity of their devices. Maintaining connectivity, identified as the best strategy in a Bayesian game, ensures cooperative behavior among validators, thereby rendering the incentive mechanism of Eth2 incentive-compatible.

Our study offers three contributions to the literature on validator participation in Eth2. \textbf{Firstly}, to equip the reader with the necessary background for understanding this paper, we consolidate key information on the reward, penalty, and cost model of Eth2 from various sources. \textbf{Secondly}, we introduce a game theoretical approach to analyzing validators' behavior and decision-making in Eth2 using a Bayesian game. Our analysis demonstrates that reaching a BNE for a block proposer\footnote{In the Eth2, a singular validator is exclusively designated as the block proposer for each discrete slot.}  and other attesters in a slot is contingent upon validators being aligned with cooperation and adherence to the rules of Eth2. Moreover, cooperation is the ex ante dominant strategy in this circumstance. \textbf{Finally}, based on the analysis of the BNE and the ex ante dominant strategy, we conclude that the incentive mechanism of Eth2 is incentive-compatible for individual block proposers and attesters within a slot. To the best of our knowledge, this represents the first attempt to use Bayesian games to analyze the impact of validator strategies on the core incentive mechanism of Eth2. In essence, the findings of this research provide significant guidance to ensure that validators can maximize their utilities by conforming rather than deviating from the Eth2 protocol. This guidance is crucial for the design and implementation of effective incentive mechanisms for validator participation in the next generation of Ethereum.

The remainder of this paper is structured as follows: Section \ref{sect_background_incentive} provides a comprehensive overview of the incentive mechanism in Eth2, laying the groundwork for subsequent analyses. Section \ref{sect_et2_bayesian_game_use_this_in_intro} introduces and elaborates on the Bayesian game model that underpins our approach. Next, in Section \ref{sect_incentive_analysis}, we conduct a detailed game analysis, highlighting key findings and implications. Subsequently, Section \ref{sect_related_work} reviews relevant literature, situating our work within the broader research context. Finally, Section \ref{sect_conclusion} summarizes our main conclusions and outlines directions for future research.

\section{Background: Incentive Mechanism in Eth2}
\label{sect_background_incentive}
To become a validator in Eth2, one needs to deposit at least 32 ETH. The validator’s total stake is referred to as the \textbf{total balance (TB)}. With the consensus process executing, validators will possibly receive rewards or penalties calculated at the end of each epoch. As a result, the \textbf{effective balance (EB)} is the latest value calculated by their $TB$, rewards and penalties\footnote{Note that the effective balance is rounded to the nearest 1 ETH increment. For example, a total balance of 31.4 and 31.6 ETH will be rounded to 31 and 32 ETH, respectively, as the effective balance.} \cite{ vitalikannotatedspec, edgington_2022}.  $EB$ represents the influence of the voting power of a validator. For example, the impact of a validator that possesses 32 ETH is more significant than another one that only has 30 ETH. However, in contrast to $TB$, a validator can have a maximum of 32 ETH as its $EB$. This limitation, designed by Eth2, ensures a balanced distribution of voting power for each validator. 

In Eth2, "finality" denotes the irreversible nature of blockchain transactions, secured by mandating a significant expenditure of ETH to alter any portion of a block \cite{vitalikannotatedspec}. Eth2 divides time into \emph{slots}, each lasting 12 seconds. These slots are organized into epochs, with each epoch comprising 32 slots and totalling approximately 6.4 minutes. To achieve "Finality" for each epoch, Eth2 utilizes "Checkpoint," composed of two valid blocks in the format $(source, target)$. $Target$ means the first block of the current epoch. If a valid block receives at least two-thirds of the $EB$ of all validators to mark it as the $target$ in a checkpoint, its status will become "justified."  Next, a "justified" valid block’s status will become "Finalized" if it acquires at least two-thirds of the $EB$ of all validators who vote the "justified" valid block as the $source$, which represents the most recent justified block \cite{buterin2017casper, buterin2020combining}. For instance, if the checkpoint of epoch zero is $(N/A, block \ \#1)$, the $target$ = block \#1 and block \#1’s status will become "justified." Subsequently,  if the checkpoint of epoch one is $(block\ \#1, block\ \#33)$, the $source$ = block \#1’s status will become "finalized" and the $target$ = block \#33’s status will transfer to "justified." 

Aside from "finality," Eth2 also uses the "Head" block to manage its fork. The "Head" block is the most recently created block verified by attesters. When a block proposer plans to propose a new block, it must link to the "Head" block. Then, the new block will be considered valid, and Eth2 can ensure that the new block is built on the latest block agreed upon by two-thirds of the attesters' stakes \cite{buterin2020combining, cassez2022formal}. In a nutshell, validators are requested to vote on $source$, $target$, and $head$ when they are designated as attesters.

In addition to their roles as block proposers or attesters, validators in Eth2 may also be selected to join the "sync committee." This committee is comprised of 512 validators who are chosen every 256 epochs (or approximately every 27 hours) and are responsible for continuously signing block headers for each slot. In Eth2, validators receive rewards or penalties based on their performance when executing different tasks. These incentives include rewards for block proposers who successfully create and include a block in the chain, rewards for attesters who correctly vote on the validity of blocks, and penalties for validators who are offline or violate the protocol. In the following sections, we will discuss the incentive mechanisms for validators in Eth2, including rewards and penalties they can receive.

\subsection{Reward, Penalty, and Cost Model of Validators in Eth2} \label{sect_reward_penalty_cost_model}
In Eth2, validators conduct their duty as attesters, joining a sync committee or proposers to get rewards. Meanwhile, validators will be punished or slashed if they do not do their duties correctly or behave maliciously \cite{buterin2020combining, buterin2020incentives}. To analyze Eth2's mechanism, we need to disassemble the rewards, punishments, and costs of Eth2. In this paper, we summarize the reward, penalty, and cost models for validators in Eth2 based on the technical specifications and research by Buterin (2020) \cite{vitalikannotatedspec}, Edgington (2023) \cite{edgington_2022} and McDonald (2020) \cite{jim2020}. The content of the sects. \ref{sect_bk_reward_model} to \ref{sect_bk_cost_model} are grounded in an analytical summary of various sources, which serve as the foundation for our models and findings.

\subsubsection{Reward Model} \label{sect_bk_reward_model}
Validators in Eth2 can receive three types of block rewards: \emph{attestation rewards}, \emph{proposer rewards}, and \emph{sync committee rewards} \cite{vitalikannotatedspec}. However, before introducing each reward, we must define two essential terms: the \textbf{effective balance increment (EBI)} and the \textbf{base reward factor (BR\_FACTOR)}.

In Eth2, the \emph{BR\_FACTOR} is set to 64 by \cite{vitalikannotatedspec}. The \emph{EBI} is a value that is set to 1 in the same source. These two values are used in the calculation of the \textbf{base reward per increment (BR)}, which is given by the following Equation \ref{eq_base_rewrd} derived from \cite{vitalikannotatedspec}:

\begin{equation}
\label{eq_base_rewrd}
\begin{split}
BR = \frac{EBI \cdot BR\_FACTOR} {\sqrt{TB}}
\end{split}
\end{equation}

Let us take a look at an example provided by \cite{edgington_2022}. Assume there are 500,000 validators, and each validator holds 32 ETH as their $\textbf{effective balance (EB)}$. In Eth2, \textbf{GWei} = $10^{-9}$ ETH, which is the smallest unit of the token. Therefore, the \emph{BR} is around 505.96 GWei, which is calculated by the following process:

\begin{equation}
\label{eq_base_rewrd_calculation_procedd}
\begin{split}
BR = \frac{1 \cdot 10^9 \cdot 64} {\sqrt{32 \cdot 10^9 \cdot 500,000}} \approx 505.96
\end{split}
\end{equation}

After obtaining the value of the \emph{BR}, as derived in the preceding equation, we may elaborate on the three types of rewards available to validators in Eth2: attestation rewards, proposer rewards, and sync committee rewards.

\paragraph{\textbf{Attestation Rewards}}
There are three types of \emph{attestation rewards} that validators can receive in each epoch if they are not selected to be proposers or join the sync committee. First, if a validator $i$ votes for the source correctly\footnote{The correctness of voting means that a validator votes for the same block as other validators possess at least two-thirds of the EB.}, it can receive the \emph{reward of the voting source} (\emph{$R_i^{AS}$}). This reward is calculated using Equation \ref{eq_reward_at_source}, where \emph{W${source}$} is the total weight of all validators that voted correctly for the source, \emph{W$_{total}$} is the total weight of all validators, \emph{EB} is the effective balance of the validator, and \emph{BR} is the base reward per increment. According to the Eth2 specification \cite{vitalikannotatedspec}, $W_{source}$ is 14, and $W_{total}$ is 64. The \emph{EB} is the amount of stake that a validator possesses. This paper assumes that each validator has an \emph{EB} of 32.

\begin{equation}
\label{eq_reward_at_source}
\begin{split}
R_i^{AS} = \frac{W_{source}}{W_{total}} \cdot EB \cdot BR
\end{split}
\end{equation}

Second, a validator $i$ could get the \emph{reward of the voting target} (\emph{R$_i^{AT}$}) in Equation \ref{eq_reward_at_target} if it votes for the target correctly. The \emph{$W_{target}$} is the weight of rewards when validators vote target correctly, which is set to 26 in \cite{vitalikannotatedspec}.

\begin{equation}
\label{eq_reward_at_target}
\begin{split} 
R_i^{AT} = \frac{W_{target}}{W_{total}} \cdot EB \cdot BR
\end{split}
\end{equation}

Third, when a validator $i$ correctly votes for the head, it is awarded the \emph{reward of the voting head} (\emph{R$_i^{AH}$}), as delineated in Equation \ref{eq_reward_at_head}. The reward weight for accurately voting the head, denoted as \emph{W$_{head}$}, is set to 14 in \cite{vitalikannotatedspec}, identical to \emph{W$_{source}$}. Given that \emph{W$_{head}$} is equal to \emph{W$_{source}$}, it follows that $R_i^{AH}$ is commensurate with $R_i^{AS}$. 

\begin{equation}
\label{eq_reward_at_head}
\begin{split} 
R_i^{AH} = \frac{W_{head}}{W_{total}} \cdot EB \cdot BR
\end{split}
\end{equation}

In summary, it can be determined that within a given slot, validator $i$ may receive the maximum \emph{attestation rewards}, represented by the sum of \emph{R$_i^{AS}$}, \emph{R$_i^{AT}$}, and \emph{R$_i^{AH}$}, which are collectively denoted as \emph{R$_i^{A}$} and calculated by Equation \ref{eq_reward_at_sum}.

\begin{equation}
\label{eq_reward_at_sum}
\begin{split} R_i^{A} = \frac{W_{source} + W_{target} + W_{head}}{W_{total}} \cdot EB \cdot BR
\end{split}
\end{equation}

\paragraph{\textbf{Sync Committee Rewards}}
If a validator $i$ completes a sync committee task in a given slot, it can receive a \emph{Sync Committee Reward} (\emph{$R_i^C$}), which is calculated using Equation \ref{eq_reward_sync_sum}. The value of \emph{$W_{sync}$}, which represents the number of reward weights earned by validators for accurately executing sync committee tasks, is set to two in \cite{vitalikannotatedspec}. 

\begin{equation}
\label{eq_reward_sync_sum}
\begin{split} 
R_i^{C} = \frac{1}{32 \cdot 512} \cdot \frac{W_{sync}}{W_{total}} \cdot N_V \cdot EB \cdot BR
\end{split}
\end{equation}

In Equation \ref{eq_reward_sync_sum}, \emph{$N_V$} is the total number of active validators in the Eth2, and the fraction $\frac{1}{32 \cdot 512}$ is because there are 512 validators in a sync committee, and each epoch has 32 slots.

\paragraph{\textbf{Proposer Rewards}} The reward given to the proposer, performed by a validator $i$, represented as $R_i^P$, can be divided into three types: presenting attesting data, called $R_i^{PA}$; synchronizing committee data, known as $R_i^{PC}$; and the tips from transactions included in its proposed block, denoted as $R_i^{PT}$. The calculation methods for $R_i^{PA}$, $R_i^{PC}$, and $R_i^{PT}$ are delineated in \cite{vitalikannotatedspec} and will be explained in this section.

The calculation for $R_i^{PA}$ is shown in Equation \ref{eq_reward_rpa}, where $N_A$ represents the total number of attesters in an epoch, and $W_{proposer}$ in Equation \ref{eq_reward_rpa} refers to the number of weights of rewards when a validator includes attesting data in its block. It is noteworthy that in \cite{vitalikannotatedspec}, $W_{proposer}$ is set to a fixed value of eight. Note that the $R_i^{PA}$ is a fraction of $R_i^A$, as illustrated in Equation \ref{eq_reward_rpa}.

\begin{equation}
\label{eq_reward_rpa}
\begin{split}
R_i^{PA} = \frac{W_{proposer}}{W_{total} - W_{proposer}} \cdot N_A \cdot R_A
\end{split}
\end{equation}

Also, the calculation method for $R_i^{PC}$ is defined in Equation \ref{eq_reward_rpc} in \cite{vitalikannotatedspec}. The constant 512 in Equation \ref{eq_reward_rpc} represents the number of validators in a sync committee.

\begin{equation}
\label{eq_reward_rpc}
\begin{split}
R_i^{PC} = \frac{W_{proposer}}{W_{total} - W_{proposer}} \cdot R_i^C \cdot 512
\end{split}
\end{equation}

Moreover, Eth2 adopts the EIP-1559 model \cite{roughgarden2020transaction} to calculate its transaction fees, which are composed of a \emph{base fee} and \emph{tips}, both paid by clients. A \emph{Base fee} is a type of dynamically adjusted fee used to balance the demand and supply of transactions in Eth2. Each Eth2 transaction is required to pay a \emph{base fee}, which is calculated based on the amount of gas consumed by a transaction. However, the \emph{base fee} is burned\footnote{The goal of burning the \emph{base fee} in Eth2 is to reduce inflation by removing Ether from circulation \cite{roughgarden2020transaction}.} rather than being delivered to block proposers to control the total amount of Ether. In contrast to the \emph{base fee}, \emph{tips} are directly delivered to a block proposer as an incentive for the block proposer to prioritize selecting this transaction into their proposed block. Hence, we can calculate $R_i^{PT}$ by Equation \ref{eq_transaction_fee_rewards} since $R_i^{PT}$ is the summation of all \emph{tips} in a proposed block. 

\begin{equation}
\label{eq_transaction_fee_rewards}
R_i^{PT} = \sum_{tx \in B_i} tip(tx)
\end{equation}

In Equation \ref{eq_transaction_fee_rewards}, $B_i$ represents the set of transactions included by proposer $i$ in its proposed block and $tip(tx)$ denotes the function of calculating \emph{tips} associated with each transaction $tx$. For detailed calculations of the \emph{base fee} and \emph{tips}, as well as their economic and technical impacts on Eth2, please refer to \cite{roughgarden2020transaction} and \cite{liu2022empirical}. Combining all components, the complete proposer reward $R_i^P$ is given by Equation \ref{eq_reward_rp_sum}.

\begin{equation}
\label{eq_reward_rp_sum}
\begin{split}
R_i^{P} = R_i^{PA} + R_i^{PC} + R_i^{PT}
\end{split}
\end{equation} 

So far, we have introduced the three types of rewards that validators could receive. It is important to note that, in Eth2, if validators exceed the time limits for executing tasks, their rewards will be reduced or taken away partially or entirely based on the extent of the delay \cite{vitalikannotatedspec, edgington_2022}. This paper assumes that validators execute their tasks \textbf{within} the allotted time frame to receive total rewards for each particular task.

\subsubsection{Punishment Model} \label{sect_bk_punishment_model}
Eth2 proposes two types of punishments. First, a validator will receive \emph{penalties} if it does not execute the tasks of being an attester, a member of the sync committee, or a block proposer. Second, if a validator violates Eth2's consensus protocol \cite{buterin2020combining}, it will be charged with \emph{slashing} and asked to eject and not be a validator anymore. In a nutshell, the \emph{penalties} are more moderate than the \emph{slashing}\footnote{Violations that lead to penalties or slashing are either automatically detected by the Eth2 consensus layer or reported by validators who collect evidence of such infractions, thereby enhancing the reliability of Eth2's detection capabilities.} \cite{vitalikannotatedspec}.

\paragraph{\textbf{Penalties}}
The \emph{penalties for attester} (\emph{$P_i^{A}$}) and \emph{penalties for sync committee} (\emph{$P_i^{C}$}) are two types of penalties that a validator $i$ might receive, as defined in \cite{vitalikannotatedspec}. If an attester votes the source incorrectly, it will get penalties equal to \emph{$R_i^{AS}$}. Also, an attester will get penalties equal to \emph{$R_i^{AT}$} if it votes the wrong target. Consequently, the maximum \emph{$P_i^{A}$} that an attester will receive is equal to \emph{$R_i^{AS}$} + \emph{$R_i^{AT}$}. Besides, there is an asymmetric punishment scheme here: an attester will not get penalties when voting incorrectly for the head, which we will discuss more in the Sect. \ref{sect_incentive_compatibility}. On the other hand, if a validator fails to execute the task of the sync committee precisely, it will get penalties, which are \emph{$P_i^C$}, equal to \emph{$R_i^{C}$}. For better understanding, we display the \emph{$P_i^{A}$} and \emph{$P_i^C$} in Equations \ref{eq_penalty_pa} and \ref{eq_penalty_pc}.

\begin{equation}
\label{eq_penalty_pa}
\begin{split}
P_i^{A} = R_i^{AS} + R_i^{AT}
\end{split}
\end{equation}

\begin{equation}
\label{eq_penalty_pc}
\begin{split}
P_i^{C} = R_i^C
\end{split}
\end{equation}

\paragraph{\textbf{Inactivity leak}} In Eth2, a mechanism known as "inactivity leak" is initialized if validators fail to confirm any checkpoints for four consecutive epochs. When "inactivity leak" starts, attesters will not receive rewards, and the inactive validators will be penalized by losing their stakes. However, blockchain proposers and sync committees will continue to receive rewards during the "inactivity leak" period. Once a checkpoint is confirmed by validators, the "inactivity leak" will end. In this paper, we utilize $P_i^{IL}$ to represent the penalty for validator $i$ during the "inactivity leak" period.

\paragraph{\textbf{Slashing}}
With the Eth2 framework, slashing is composed of three distinct types and forms a critical part of the incentive mechanism, specifically designed to deter and penalize malicious behaviour among validators. Validators are subjected to slashing only when they engage in actions that violate the established rules applicable to attesters or proposers \cite{vitalikannotatedspec}.
First, the \emph{initial slashing} (\emph{$S_i^{I}$}) will be applied to validators who violate Eth2's rules. In \cite{vitalikannotatedspec}, the \emph{$S_i^{I}$} is set to $\frac{1}{6}$. Second, a validator $i$ who is slashing will continue to receive the \emph{correlation of slashing} (\emph{$S_i^{C}$}). Third, before being ejected, \emph{attestation of slashing} (\emph{$S_i^{A}$}) continues being applied to the validator $i$ who is being slashed. However, it is crucial to highlight that this research is solely focused on the incentives for Eth2 validators and deliberately omits consideration of the slashing penalties for validators in breach of Eth2's protocol. Although related, slashing is distinct from the incentive mechanisms that motivate validators to participate in the Eth2 consensus protocol, which is the main focus of this study. Therefore, it is \textbf{outside the scope of this study}.

\subsubsection{Cost Model} \label{sect_bk_cost_model}
To the best of our knowledge, most of the literature discusses the execution costs in Eth2, such as the amount of \emph{gas} required for certain commands in Solidity, rather than explicitly addressing the cost analysis of validators in Eth2. Therefore, we conducted our analysis based on our understanding and presented a cost model that consists of the following components:

\paragraph{\textbf{Equipment Cost}}
A reliable infrastructure is necessary for becoming a Validator in Eth2. This infrastructure includes computer hardware, network bandwidth, operating systems, and application software, all of which incur associated costs. To account for these costs for a validator $i$, we refer to \cite{jim2020} to define the annual equipment cost $C_i^{AE}(y)$, as we denote $y$ as the number of years that agent $i$ serves as a validator in Eth2. As shown in Equation \ref{eq_cost_equipment}, $C_i^{AE}(y)$ consists of three major components: $C_i^S$ represents the initial cost of setting up the infrastructure, such as evaluating, installing, and testing the hardware and software, which is a one-time expense. $C_i^I(y)$ reflects the purchase cost of software and hardware, recognized as depreciation over their useful life. $C_i^O(y)$ covers the costs of maintaining the infrastructure, including irregular software and hardware upgrades, updates, or replacements. For a more accurate reflection of financial commitment over time, $C_i^{AE}(y)$ will be amortized over $y$.

\begin{equation}
\label{eq_cost_equipment}
C_i^{AE}(y) = \frac{C_i^S + C_i^I(y) + C_i^O(y)}{y}
\end{equation}

\paragraph{\textbf{Participation Cost}} The participation cost for becoming a validator $i$ in Eth2 is set to 32 ETH, denoted as $C^P(y)$, as cited in \cite{vitalikannotatedspec}. Here, we denote the annual participation cost as $C_i^{AP}(y)$. Although agent $i$ can withdraw staked ETH, we still consider this 32 ETH as an ongoing cost for a validator. It is amortized over the number of years $y$ that agent $i$ serves as a validator in Eth2, as shown in Equation \ref{eq_cost_participation}.

\begin{equation}
\label{eq_cost_participation}
C_i^{AP}(y) = \frac{C^P(y)}{y}
\end{equation}

\paragraph{\textbf{Execution Cost}}
This cost type refers to the energy consumption a validator $i$ might incur when executing its tasks, including offering attestations, joining a sync committee, or proposing blocks. We denote these three components of annual execution cost ($C_i^{AX}(y)$) as $C_i^A(y)$, $C_i^Q(y)$, and $C_i^B(y)$, respectively. As a result, we can calculate $C_i^{AX}(y)$ using Equation \ref{eq_cost_computation}. The $y$ in Equation \ref{eq_cost_computation} represents the number of years $y$ that agent $i$ serves as a validator in Eth2, and this period is used to amortize these costs.

\begin{equation}
\label{eq_cost_computation}
C_i^{AX}(y) = \frac{C_i^A(y) + C_i^Q(y) + C_i^B(y)}{y}
\end{equation} 

In sum, the total annual cost ($C_i^{AT}(y)$) that a validator $i$ incurs is calculated as follows:

\begin{equation}
\label{eq_cost_total}
C_i^{AT}(y) = C_i^{AE}(y) + C_i^{AP}(y) + C_i^{AX}(y)
\end{equation} 

We utilize $C_i^{AT}(y)$ to represent the annual total cost of a validator $i$. Knowing that an Eth2 epoch consists of 32 slots, with each slot lasting 12 seconds, we can estimate the cost of a validator $i$ in each epoch by calculating the number of epochs per year using the following equation:

\begin{equation}
\label{eq_numbers_of_epoch_per_year}
N_{epochs} = \frac{60 \cdot 60 \cdot 24 \cdot 365}{12 \cdot 32}
\end{equation}

Given the total annual cost $C_i^{AT}(y)$, the cost per epoch of a validator $i$, denoted by $C_i^{EP}(y)$, can be calculated by dividing the total annual cost by the number of epochs in a year, as shown in:

\begin{equation}
\label{eq_cost_of_each_epoch}
C_i^{EP}(y) = \frac{C_i^{AT}(y)}{N_{epochs}}
\end{equation}

Thus far, we have provided fundamental knowledge of the reward, punishment, and cost model of Eth2. In the subsequent sections, specifically Sections \ref{sect_et2_bayesian_game_use_this_in_intro} and \ref{sect_incentive_analysis}, we will selectively apply key equations from the range \ref{eq_base_rewrd} through \ref{eq_cost_of_each_epoch} to further our analysis. These selected equations are crucial for understanding the rewards, penalties, or slashing, as well as the conditions that trigger these outcomes for Eth2 validators. They serve as indispensable cornerstones for the arguments we present in this paper.
\section{The Eth2 Bayesian Game}
\label{sect_et2_bayesian_game_use_this_in_intro}
In the Eth2 protocol, validators are required to perform various tasks that maintain the security and integrity of the system, such as proposing blocks, attesting to blocks, or joining a sync committee. In this work, we focus on modeling the incentives of validators during \textbf{one independent slot}\footnote{We consider each slot as an independent Bayesian game, which is reasonable since each slot has a specific start and end time, and each validator may face different situations from its counterparts, possibly leading to different strategic choices.} in Eth2. Using a Bayesian game framework, we aim to gain insight into their possible actions within the protocol and design mechanisms to encourage positive behaviors while discouraging negative ones.

\subsection{Game Model} \label{sect_game_model}
We model the behavior of multiple agents in the Eth2 protocol using a Bayesian game framework. A specific type classifies each agent and has a set of available actions. At the start of the game, agents are unaware of their types, but they can observe them once the game begins. Agents then take actions simultaneously without knowing other agents' types. The game determines the payoff or "utility" for each agent based on all agents' types and actions. Our goal is to scrutinize the game and ascertain the existence of a BNE for all agents when they choose to cooperate. Further, this cooperation strategy constitutes a dominant strategy for all agents involved. Note that we focus on the interactions between validators when their types are either block proposers or attesters. This study concentrates on the interactions between validators and attesters in a single slot.

\begin{definition}[Bayesian Game \cite{shoham2008multiagent}]\label{def_bayesiangame_ethsecurity} A Bayesian game is built on top of a tuple $(N, A, \Theta, p, u)$, where:
 \begin{itemize}
     \item $N$ is a set of agents;
     \item $A$ = $A_1 \times \dots \times A_n$, where A$_i$ is a set of actions available to player $i$;
     \item $\Theta = \Theta_1 \times \dots \times \Theta_n$, where $\Theta_i$ is the type space of player $i$;
     \item $p: \Theta \mapsto [0,1]$ is a common prior over types; and
     \item $u$ = $(u_1,\dots\,u_n)$, where $u_i:A \times \Theta \mapsto \mathbb{R}$ is the utility function for player $i$.
 \end{itemize}
\end{definition}

\paragraph{\textbf{Agents and Actions}} 
In this paper, we model validators in Eth2 as rational agents who seek to maximize their profit, represented by their utility function. We consider a game where each agent has a set of available actions, specifically two possible actions: (i) cooperate, $ C $, which keeps the agent's device online, and (ii) deviate, $ D $, which turns the device offline at least one epoch to potentially save costs. The use of cloud or centralized cryptocurrency exchanges, which could also affect costs, is not considered in this model. The set of actions available to player $ i $ is denoted by $ A_i = \{C, D\} $, and an action profile for all agents is represented by $ (a_1, a_2, \ldots, a_n) $, where $ a_i $ indicates the action taken by player $ i $, and $ n $ is the total number of agents with $ n \geq 2 $. Actions taken by all agents except for player $ i $ are denoted by $ a_{-i} $, thus an action profile can be succinctly expressed as $ (a_i, a_{-i}) $.

\paragraph{\textbf{Types}} Recall that an epoch in the Eth2 protocol consists of 32 slots, each with its block committee (note that this is different from the sync committee). The block committee randomly selects a validator in each slot as the block proposer. Therefore, a validator could be either a block proposer or an attester during an epoch. There is another type, which is a member of the sync committee. However, we will not consider the type of validators when they are a member of the sync committee since the probability of being selected for this role is very low. We focus on modeling the interactions between block proposers and attesters and gain insight into their possible actions and strategies within the protocol.

With this context in mind, agents are divided into two types in our model setup. We represent the potential types of agent $i$ in a single slot by $\theta_i \in \{B, A \}$ can be either block proposer ($B$) or attester ($A$). This notation is employed to explore all possible strategic configurations within our theoretical framework. One critical point is that our analysis focuses on a one-shot game within a single slot without considering the relationships or interactions between consecutive slots. Despite this, our approach is applicable to any slot in the Eth2 protocol, treating each slot as an independent one-shot game scenario.

\paragraph{\textbf{Prior probabilities}}
We denote the probability of being selected as a block proposer or an attester as $p_B$ and $p_A$, respectively. Although agents are aware of these probabilities before the game starts, they cannot observe their types until the game begins.

\paragraph{\textbf{Utility functions}} We use $u_i$ to represent the utility function that maps action profiles, denoted by $a = (a_1, \dots, a_n)$, and types, represented by $\theta = (\theta_1, \dots, \theta_n)$, to real-number payoffs. In this paper, we consider a game with multiple validators, where each validator can be one of two types, $\theta \in \{B, A\}$, and each type has two distinct actions, $a \in \{C, D\}$. Hence, we could utilize $u_i(a, \theta)$ to calculate the utility of agent $i$.

\subsection{Strategies and Equilibria} \label{sect_game_model_equilibria}
A strategy indicates the contingency plan of each agent $i$ regarding how they play the game. In this study, we use $\Delta(A_i)$ to represent the collection of all possible probability distributions over $A_i$, where $A_i$ is the set of actions available to agent $i$. The function $s_i : \Theta \rightarrow \Delta(A_i)$ is a mapping from agent $i$'s type to the probability distribution of agent $i$'s actions. $S_i$ is the complete collection of strategies available to agent $i$. Then, $s_i(a_i | \theta_i)$ is the probability that agent $i$ will choose action $a_i$ while following strategy $s_i$, given that their type is $\theta_i$. Finally, the vector $s = (s_1, \ldots, s_n) \in S$ represents the strategies for all agents, where $S = S_1 \times \ldots \times S_n$ is the set of all possible strategy profiles, with $n \geq 2$ indicating the number of agents.

\paragraph{\textbf{Expected Utilities}}
As discussed previously, $P_B$ and $P_A$ represent the probabilities of a validator $i$ being selected as a block proposer ($B$) or an attester ($A$), respectively. Building on the theory of conditional expected utility \cite{luce1971conditional}, we construct the expected utility for agent $i$ in our model as follows:

\begin{equation} \label{expected_utilities_bayesianGame_final}
\begin{split}
EU_i(s) &= \mathbb{E}_{a, \theta} [u_i(s, \theta)] \\ 
&= P_B \cdot \left(\mathbb{E}_{a, \theta}[u_i(a, \theta) \mid \theta_i = B] \right) \\
&+ P_A \cdot \left(\mathbb{E}_{a, \theta}[u_i(a, \theta) \mid \theta_i = A] \right)
\end{split}
\end{equation}

Equation \ref{expected_utilities_bayesianGame_final} is utilized to calculate the expected utility of validator $i$ without knowing its types before the game commences. Leveraging Equation \ref{expected_utilities_bayesianGame_final}, we model the uncertainty of types and actions within Bayesian games. Prior to the game's commencement, agents are aware of the probabilities denoted by $P_B$ and $P_A$, reflecting the likelihood of being selected as a block proposer or an attester, respectively, based on common prior knowledge.

To more comprehensively capture the nuanced strategic interactions among validators of varying types, it is essential to introduce a more detailed construct of expected utilities. Therefore, we further refine the expected utility calculation to consider specific combinations of strategies and types for agents $i$ and $-i$ as shown below:

\begin{equation} \label{eq_eu_i_simplified_refined_adjusted_for_types_certain_match}
\begin{split}
EU_i(s_i, s_{-i}, \theta_i, \theta_{-i}) &= p(\theta_{i}, \theta_{-i}) \cdot (u_i(a_i, a_{-i}, \theta_i, \theta_{-i}) - C_i^{EP}(y)) \\
\end{split}
\end{equation}

In Equation \ref{eq_eu_i_simplified_refined_adjusted_for_types_certain_match}, $p(\theta_{i}, \theta_{-i})$ represents the probability of agents $i$ and $-i$ having specific types $\theta_i$ and $\theta_{-i}$, quantifying the likelihood that this type combination occurs. $u_i(a_i, a_{-i}, \theta_i, \theta_{-i})$ is the utility function for agent $i$, measuring the payoff that agent $i$ receives from choosing action $a_i$ while other agents choose $a_{-i}$, given their respective types. Additionally, $a_{-i}$ represents the combination of actions chosen by agents other than $i$. For instance, in an Eth2 slot, it is possible that some agents within the group of agents $-i$ may choose to cooperate ($C$), while others may decide to deviate ($D$), reflecting the diverse strategic decisions possible within our game context.

\begin{definition} \label{def_x_bc_and_x_bd}
Let $x^{BC} = 1$ denote the condition that "there exists a unique validator among agents $-i$, denoted by $j$, such that $\theta_j = B$ and this agent executes $a_j = C$." If no such validator exists, let $x^{BC} = 0$. Similarly, let $x^{BD} = 1$ denote the condition that "there exists a unique validator among agents $-i$, denoted by $k$, such that $\theta_k = B$ and this agent executes $a_k = D$." If no such validator exists, let $x^{BD} = 0$.
\end{definition}

\begin{definition} \label{def_gamma}
$\gamma$ denotes the proportion of agents $-i$ choosing action $D$. Specific thresholds related to $\gamma$ influence agents' strategic decisions within the system, particularly triggering Eth2's "inactivity leak" mechanism when certain conditions are met. This mechanism affects the overall system's resilience and agents' behaviors under non-participation or reduced activity scenarios.
\end{definition}

Moreover, in our analysis in Section \ref{sect_incentive_analysis}, we examine unique conditions across various cases. Utilizing the theory of conditional expected utility \cite{luce1971conditional}, which allows us to model the utility of agents under specific scenarios, we define the expected utility of agent $i$ within our game theoretical framework. This method helps illustrate how different scenarios, represented by $\Psi$, affect the decision-making processes of the agents. The expected utility, conditional on these scenarios, is detailed in Equation \ref{eq_eu_i_simplified_certain_condition} as follows:

\begin{equation} \label{eq_eu_i_simplified_certain_condition}
\begin{split}
&EU_i(s_i, s_{-i}, \theta_i, \theta_{-i} \mid \Psi) = \\ 
&p(\theta_{i}, \theta_{-i} \mid \Psi) \cdot (u_i(a_i, a_{-i}, \theta_i, \theta_{-i} \mid \Psi) - C_i^{EP}(y))
\end{split}
\end{equation}

where $(u_i(a_i, a_{-i}, \theta_i, \theta_{-i} \mid \Psi)$ is the conditional utility, as proposed by \cite{Gottinger_1973}. Also, $\Psi$ represents different conditions in various sub-cases in Section \ref{sect_incentive_analysis} as follows:
\begin{itemize}
    \item $\Psi$ represents scenarios where one of the agents $-i$ acts as a block proposer, and all other agents $-i$ act as attesters. This includes (1) \textbf{$x^{BD} = 1$}, indicating that there exists a unique validator among agents $-i$, as defined in Definition \ref{def_x_bc_and_x_bd}, who executes action $D$, and (2) \textbf{$x^{BC} = 1$}, indicating that there exists a unique validator among agents $-i$, as defined in Definition \ref{def_x_bc_and_x_bd}, who executes action $C$.
    \item $\Psi$ represents scenarios where the proportion $\gamma$ of agents $-i$ choosing action $D$ is below or above a certain threshold. Specifically, this includes conditions where $\gamma \geq d$ or $\gamma < d$, with $d$ defined as a variable threshold necessary to trigger the "inactivity leak" in Eth2.
\end{itemize}

\paragraph{\textbf{Best Response and Bayesian Nash Equilibrium}} After defining the expected utility of an agent $i$ in Equation \ref{eq_eu_i_simplified_refined_adjusted_for_types_certain_match}, we can further define the set of agent $i$'s best responses to the strategy profiles of $s_{-i}$ as:

\begin{equation} \label{eq_bestResponse}
    BR_i(s_{-i}) = \arg \max_{s_i \in S_i} EU_i(s_i, s_{-i})   
\end{equation}

Equation \ref{eq_bestResponse} implies that each agent $i$ will attain the maximum expected utility by selecting their best response. In a Bayesian game, agents formulate their strategies based on their types and prescribe these strategies to be followed before the game starts. After the game commences, each agent $i$ discovers their type and acts accordingly, based on their pre-established contingency plan (strategies). A BNE is achieved when each agent's strategy constitutes the best response to the strategies chosen by the other agents. Mathematically, a strategy profile $s^*$ is a BNE if and only if $s^*_i \in BR_i(s^*_{-i})$ for every agent $i$. In this paper, our goal is to demonstrate that choosing action $C$ constitutes a BNE for validators within a single slot, especially considering validators that fall into two distinct categories based on their types. Additionally, we show that $s_i^*(C)$ is the dominant strategy of agent $i$, irrespective of their type, which benefits us in further concluding that Eth2's incentive mechanism is incentive compatible. 

\section{Incentive Analysis of validators in Eth2 - Methodology and Findings} \label{sect_incentive_analysis}
This section presents the primary methodology and analysis results of this paper. Specifically, we formulated the utilities of agents and applied Bayesian game analysis to obtain three key findings. First, we consider the original reward scheme of Ethereum 2.0 and demonstrate that cooperation will emerge as a BNE strategy for validators of two distinct categories in our setting under this reward scheme, regardless of their type (See Theorem \ref{theorem_bne_cooperate_eth}). Furthermore, we show that under this reward scheme, cooperation becomes an \emph{ex ante dominant strategy} for validators of any type (See Theorem \ref{theorem_bne_cooperate_eth}). Lastly, based on the premise that cooperation is an ex ante dominant strategy, we infer that Ethereum 2.0's incentive mechanism ensures incentive compatibility for both single block proposers and attesters in Eth2 (See Corollary \ref{corollary_incentive_compatibility_eth}).

\subsection{Analysis of Incentive Compatibility for Validators in Eth2} \label{sect_incentive_compatibility}
Validators in Eth2 are subject to varying rewards or penalties and incur different costs depending on their actions \cite{vitalikannotatedspec, edgington_2022}. A strategic behavior observed among validators is taking their devices offline to reduce operational costs. This paper employs the Bayesian game model we developed to demonstrate that taking devices offline is not the best strategy for validators. Before delving into the analysis, let us revisit the Bayesian game framework defined in Section \ref{sect_game_model}. We consider a network of validators, where each validator belongs to one of two types, represented by $\theta_i \in \{B, A\}$, denoting the roles of block proposer and attester, respectively. Each validator, irrespective of their type, can undertake actions $a_i \in \{C, D\}$, corresponding to cooperating (keeping the device online) or deviating (taking the device offline). In this paper, we focus on validators' strategies of choosing between keeping their devices online or offline. Analyzing more complex strategies represents a promising direction for future research.

\subsubsection{Validators' Device Offline Behavior Equilibrium Analysis}
Using the Bayesian game model described in Section \ref{sect_game_model}, we investigate whether it is the best response for validators to take their devices offline under the current Eth2 reward and penalty scheme in any given slot. In our analysis, we consider the general scenario where validators operate their own verification nodes rather than through staking pools, such as those provided by large exchanges. In the context of using centralized staking pools provided by large exchanges, the cost incurred by validators refers to the fees charged by these exchanges for validation services rather than the direct costs associated with setting up and maintaining validator client software and hardware.

\begin{theorem}\label{theorem_bne_cooperate_eth}
Let $s^*$ denote a strategy profile in the Bayesian game model for an arbitrary number of validators $n$, where each validator $i$ can be either a block proposer ($\theta_i = B$) or an attester ($\theta_i = A$). Assuming that all validators utilize a standard Ethereum client without participating in centralized delegation pools, thereby facing the same cost per epoch of a validator, $C_i^{EP}(y)$, the strategy $s_i^*(C \mid \theta_i)$ that encourages validators to cooperate by keeping their devices online not only constitutes a BNE but also represents an ex ante dominant strategy under the Eth2's incentive mechanism.
\end{theorem}

\begin{proof}
In the Bayesian game model we have defined, validators are categorized into two types: $\theta_i \in \{B, A\}$. Each type has two possible actions, denoted by $a_i \in \{C, D\}$. We subsequently detail our proof through four distinct cases, examining how agent $i$, representing one of the validators, adapts strategies in response to the actions of all other validators, denoted by agents $-i$. 

Across Cases 1 to 3, we rigorously analyze the variations in the expected utility of agent $i$ as it responds to the strategies employed by agents $-i$. Case 4 further explores the impact of agent $i$ employing mixed strategies on its expected utility, which also depends on the strategies of agents $-i$, thereby ensuring the rigor and thoroughness of our analysis. Furthermore, it is critical to note that in Cases 2 to 4, agents $-i$ may adopt either pure or mixed strategies. Specifically, in Case 1, agents $-i$ employ pure strategies, consistently choosing either $C$ or $D$, while in subsequent cases, some agents $-i$ may vary their choices between $C$ and $D$ from one slot to another. \\

\textbf{Case 1:} \, $\forall \theta_{-i} \in \{B, A\}, \, a_{-i} = C$

\textbf{Hypothesis:} Agent $i$ plays a pure strategy, choosing either $C$ or $D$. Additionally, all agents $-i$ play $C$ as their pure strategies. Consequently, in all sub-cases of Case 1, we uniformly set $a_{-i}$ in Equation \ref{eq_eu_i_simplified_refined_adjusted_for_types_certain_match} to $C$. Specifically, in Sub-case 1.1, this results in Equations \ref{expected_utilities_subcase_1_1_01} and \ref{expected_utilities_subcase_1_1_02}; in Sub-case 1.2, it leads to Equations \ref{expected_utilities_subcase_1_2_01} to \ref{expected_utilities_subcase_1_2_04}.

\textbf{Analysis:} In this case, since our goal is to analyze the impact on agent $i$ when choosing either $a_{i} = C$ or $a_{i} = D$ as their pure strategy under different circumstances, and given that all other agents $-i$ consistently choose $a_{-i} = C$, we examine two sub-cases: agent $i$ will be either a block proposer ($\theta_i = B$) or an attester ($\theta_i = A$).

\paragraph{\emph{\textbf{Sub-case 1.1:}}} \label{section_subcase_1-1} $\theta_i = B$, which means agent $i$ is the block proposer, with all agents $ -i $ playing $C$. 

In this sub-case, we can disregard the scenario where $\theta_{-i} = B$, as the Eth2 protocol stipulates that only one validator is selected as the block proposer per slot. Also, with this rule in mind, we will utilize the following two equations to conduct our analysis. First, Equation \ref{expected_utilities_subcase_1_1_01} represents the situation when $\theta_{-i} = A$ and agent $i$ plays $C$.

\begin{equation} \label{expected_utilities_subcase_1_1_01}
\begin{split}
EU_i(C, C, B, A) & = p(B, A) \cdot (u_i(C, C, B, A) - C_i^{EP}(y)) \\
&=  p(B, A) \cdot (R_i^P - C_i^{EP}(y))
\end{split}
\end{equation}

where $R_i^P$ and $C_i^{EP}(y)$ denote the reward of the block proposer and the execution cost of the block proposer, respectively, as introduced in Sections \ref{sect_bk_reward_model} and \ref{sect_bk_cost_model}, separately. 

Second, Equation \ref{expected_utilities_subcase_1_1_02} represents when $\theta_{-i} = A$ and the agent $i$ plays $D$, with all agents $ -i $ playing $C$.

\begin{equation} \label{expected_utilities_subcase_1_1_02}
\begin{split}
EU_i(D, C, B, A) & = p(B, A) \cdot (u_i(D, C, B, A) - C_i^{EP}(y)) \\
&= p(B, A) \cdot (-~C_i^{EP}(y))
\end{split}
\end{equation}

where $u_i(D, C, B, A) = 0$ because Eth2 does not penalize the block proposer if they fail to execute their task.

According to examples in \cite{edgington_2022, jim2020}, we know $R_i^P > C_i^{EP}(y)$. Consequently, based on Equations \ref{expected_utilities_subcase_1_1_01} and \ref{expected_utilities_subcase_1_1_02}, we can observe that \linebreak $EU_i(C, C, B, A) > EU_i(D, C, B, A)$, which indicates choosing $C$ is the best response of agent $i$ in this sub-case. \\

\textbf{Sub-case 1.2:} \label{sect_sub_case_1_2} $\theta_i = A$, which means agent $i$ is an attester, with all agents $-i$ playing $C$.

\textbf{Scenario 1.2.1:} The condition $x^{BC} = 1$ holds, indicating that a unique validator among agents $-i$ executes $s_{-i}(C)$ with $\theta_{-i} = B$. In this scenario, when $x^{BC} = 1$ and agent $i$ plays $C$, we can obtain:

\begin{equation} \label{expected_utilities_subcase_1_2_01}
\begin{split}
& EU_i(C, C, A, A \mid x^{BC} = 1) = \\ 
& p(A, A \mid x^{BC} = 1) \cdot (u_i(C, C, A, A \mid x^{BC} = 1) - C_i^{EP}(y)) =\\
& p(A, A \mid x^{BC} = 1) \cdot (R_i^A - C_i^{EP}(y))
\end{split}
\end{equation}

where $p(A, A \mid x^{BC} = 1)$ represents the conditional probability of that when $\theta_i = A$ and except for one of the agents $-i$'s type $= B$, all other agents $-i$'s type = A. 

When $x^{BC} = 1$ and agent $i$ plays $D$. We can obtain:

\begin{equation} \label{expected_utilities_subcase_1_2_03}
\begin{split}
&EU_i(D, C, A, A \mid x^{BC} = 1) = \\
&p(A, A \mid x^{BC} = 1) \cdot (u_i(D, C, A, A \mid x^{BC} = 1) - C_i^{EP}(y)) =\\
&p(A, A \mid x^{BC} = 1) \cdot (-P_i^A - C_i^{EP}(y))
\end{split}
\end{equation}

\textbf{Scenario 1.2.2:} The condition $x^{BC} = 0$ holds, indicating all agents -i are attesters.   

In this scenario, when $\theta_{-i} = A$ and agent $i$ plays $C$. We can obtain:

\begin{equation} \label{expected_utilities_subcase_1_2_02}
\begin{split}
EU_i(C, C, A, A) & = p(A, A) \cdot (u_i(C, C, A, A) - C_i^{EP}(y)) \\
&=  p(A, A) \cdot (R_i^A - C_i^{EP}(y))
\end{split}
\end{equation}

when $\theta_{-i} = A$ and agent $i$ plays $D$. We can obtain:

\begin{equation} \label{expected_utilities_subcase_1_2_04}
\begin{split}
EU_i(D, C, A, A) & = p(A, A) \cdot (u_i(D, C, A, A) - C_i^{EP}(y)) \\
&=  p(A, A) \cdot (-P_i^A - C_i^{EP}(y))
\end{split}
\end{equation}

Note that in Eth2, at least one validator is required to be selected as the blockchain proposer in a slot. Consequently, $p(A, A)$ in Equations \ref{expected_utilities_subcase_1_2_02} and \ref{expected_utilities_subcase_1_2_04} is equal to 0, resulting in the values of Equations \ref{expected_utilities_subcase_1_2_02} and \ref{expected_utilities_subcase_1_2_04} being equal to 0. Moreover, it is similarly shown in \cite{edgington_2022, jim2020} that $R_i^A > C_i^{EP}(y)$. Therefore, based on Equations \ref{expected_utilities_subcase_1_2_01} to \ref{expected_utilities_subcase_1_2_04}, we can readily conclude that $s_i^*(C)$ is the best response of agent $i$ with $\theta_i = A$, $\forall \theta_{-i} = B$ or $A$. \\

In summary, the analytical results from sub-cases 1.1 and 1.2 consistently demonstrate that if all agents $-i$ choose $C$, regardless of whether they are type $B$ or $A$, then selecting $C$ is the best response for agent $i$ in Case 1, irrespective of the type of agents $-i$. Thus far, in Case 1, we discuss the impact of agent $i$'s expected utility when agents $-i$ only chooses $C$. Following this, in Cases 2 to 4, we aim to consider the influence on agent $i$'s expected utility under the circumstance when agents $-i$ opts for $D$. \\

\textbf{Case 2:} \label{para_case_02_new} $\theta_i = B$, which means that agent $i$ is the block proposer.

\textbf{Hypothesis:} Any agent $i$ still plays a pure strategy of either $a_i = C$ or $a_i = D$. However, all agents $-i$ play mixed strategies\footnote{Note that agents $-i$ may still play a pure strategy, since a pure strategy is a special case of a mixed strategy.}.

\textbf{Analysis:} This case analyzes the impact of the "inactivity leak" from some agents $-i$ playing $s_{-i}(D)$ on agent $i$ with $\theta_i = B$. We will explore this through two sub-cases. \\

\textbf{Sub-case 2.1}: All $\theta_{-i} = A$, and less than $\frac{1}{3}$ of agents $-i$ choose to play $s_{-i}(D)$.

Given that the reward of a block proposer in Eth2 varies depending on the number of attesters, and considering that fewer than $\frac{1}{3}$ of agents $-i$ choose $D$, agent $i$'s reward, while potentially reduced, remains guaranteed. This assumption is based on the Eth2 reward mechanism, where a block proposer is always rewarded, albeit the total reward may decrease with fewer attesters. Assuming $\gamma$, which denotes the proportion of agents $-i$ choosing $s_{-i}(D)$, is less than $\frac{1}{3}$ and agent $i$ opts for $C$, we can employ the following equation to calculate the utility of agent $i$:

\begin{equation} \label{expected_utilities_subcase_2_1_01}
\begin{split}
& EU_i(C, s_{-i}, B, A \mid \gamma < \frac{1}{3}) \\
& = p(B, A \mid \gamma < \frac{1}{3}) \cdot (u_i(C, a_{-i}, B, A \mid \gamma < \frac{1}{3}) - C_i^{EP}(y)) \\
& = p(B, A \mid \gamma < \frac{1}{3}) \cdot (R_i^P - C_i^{EP}(y))
\end{split}
\end{equation}

In this equation, $a_{-i}$ represents the actions taken by all agents other than $i$ across all scenarios in Cases 2 to 4. The actions within $a_{-i}$ can vary, reflecting different combinations of $C$ and $D$ as chosen by the attesters. This model allows us to account for the diversity of strategies in the system and better reflects the complexity of the interaction among validators under the specified conditions. Note that in the following equations of different sub-cases, the definition of $\gamma$ remains the same. Then, when agent $i$ plays $D$, we can obtain:

\begin{equation} \label{expected_utilities_subcase_2_1_02}
\begin{split}
& EU_i(D, s_{-i}, B, A \mid \gamma < \frac{1}{3}) \\
& = p(B, A \mid \gamma < \frac{1}{3}) \cdot (u_i(D, a_{-i}, B, A \mid \gamma < \frac{1}{3}) - C_i^{EP}(y)) \\
&=  p(B, A \mid \gamma < \frac{1}{3}) \cdot (-~C_i^{EP}(y))
\end{split}
\end{equation}

In Equation \ref{expected_utilities_subcase_2_1_01}, agent $i$ receives $R_i^P$ upon successfully completing its task as a block proposer. Notably, Equation \ref{expected_utilities_subcase_2_1_02} highlights a unique aspect of Eth2's incentive mechanism: block proposer, even if they opt out of their assigned task (representing by choosing D), are not subjected to direct penalties. Therefore, based on equations \ref{expected_utilities_subcase_2_1_01} and \ref{expected_utilities_subcase_2_1_02}, we can observe, in sub-case 2.1, $s_i^*(C)$ is still the best response of agent $i$. \\

\textbf{Sub-case 2.2}: All $\theta_{-i} = A$, with at least $\frac{1}{3}$ of agents $-i$ choosing to play $s_{-i}(D)$. Since the “inactivity leak” in Eth2 is triggered when no new checkpoints are verified for four consecutive epochs, we need to consider two scenarios in this sub-case. \\

\textbf{Scenario 2.2.1:} The "inactivity leak" is triggered. In this scenario, agent $i$ still can receive $R_i^p$ since Eth2 only stops rewarding attesters in the "inactivity leak" period. Assume that $\gamma \geq \frac{1}{3}$ and agent $i$ plays $C$, we can utilize the following equation to compute the utility of agent $i$:

\begin{equation} \label{expected_utilities_subcase_2_2_01}
\begin{split}
& EU_i(C, s_{-i}, B, A \mid \gamma \geq \frac{1}{3}) \\ 
& = p(B, A \mid \gamma \geq \frac{1}{3}) \cdot (u_i(C, a_{-i}, B, A \mid \gamma \geq \frac{1}{3}) - C_i^{EP}(y)) \\
&=  p(B, A \mid \gamma \geq \frac{1}{3}) \cdot (R_i^p - C_i^{EP}(y))
\end{split}
\end{equation}

Assume that $\gamma \geq \frac{1}{3}$ and agent $i$ play $D$, we can utilize the following equation to compute the utility of agent $i$:

\begin{equation} \label{expected_utilities_subcase_2_2_02}
\begin{split}
& EU_i(D, s_{-i}, B, A \mid \gamma \geq \frac{1}{3}) \\ 
& = p(B, A \mid \gamma \geq \frac{1}{3}) \cdot (u_i(D, a_{-i}, B, A \mid \gamma \geq \frac{1}{3}) - C_i^{EP}(y)) \\
&=  p(B, A \mid \gamma \geq \frac{1}{3}) \cdot (-~C_i^{EP}(y))
\end{split}
\end{equation}

where $(u_i(D, a_{-i}, B, A \mid \gamma \geq \frac{1}{3})$ = 0 in Equation \ref{expected_utilities_subcase_2_2_02} since agent $i$ does not execute its task in this scenario. \\

\textbf{Scenario 2.2.2:} The "inactivity leak" is not triggered. In this scenario, we can employ the following equations to calculate the utilities of agent $i$: 

\begin{equation} \label{expected_utilities_subcase_2_2_03}
\begin{split}
& EU_i(C, s_{-i}, B, A \mid \gamma \geq \frac{1}{3}) \\
& = p(B, A \mid \gamma \geq \frac{1}{3}) \cdot (u_i(C, a_{-i}, B, A \mid \gamma \geq \frac{1}{3}) - C_i^{EP}(y)) \\
&=  p(B, A \mid \gamma \geq \frac{1}{3}) \cdot (R_i^p - C_i^{EP}(y))
\end{split}
\end{equation}

\begin{equation} \label{expected_utilities_subcase_2_2_04}
\begin{split}
& EU_i(D, s_{-i}, B, A \mid \gamma \geq \frac{1}{3}) \\
& = p(B, A \mid \gamma \geq \frac{1}{3}) \cdot (u_i(D, a_{-i}, B, A  \mid \gamma \geq \frac{1}{3}) - C_i^{EP}(y)) \\
&=  p(B, A \mid \gamma \geq \frac{1}{3}) \cdot (-~C_i^{EP}(y))
\end{split}
\end{equation}

In Equation \ref{expected_utilities_subcase_2_2_03}, although more than $\frac{1}{3}$ of agents $-i$ with $\theta_{-i} = A$ play $s_{-i}(D)$, which will reduce the agent $i$'s reward ($R_i^p$), references \cite{edgington_2022,jim2020} indicate that in Equation \ref{expected_utilities_subcase_2_2_03}, $R_i^p$ is indeed greater than $C_i^{EP}(y)$. Furthermore, based on the Eth2 protocol, in Equation \ref{expected_utilities_subcase_2_2_04}, it is established that $u_i(D, a_{-i}, B, A  \mid \gamma \geq \frac{1}{3}) = 0$ since the agent $i$ does not execute its task of a block proposer, reflecting that agent $i$ incurs no penalties for not fulfilling its block proposer responsibilities when $\gamma \geq \frac{1}{3}$. 

Upon closer examination, we find that \ref{expected_utilities_subcase_2_2_01} and \ref{expected_utilities_subcase_2_2_03} are identical, just as  \ref{expected_utilities_subcase_2_2_02} and \ref{expected_utilities_subcase_2_2_04} share the same form. Therefore, based on these equations, we conclude that $s_i^*(C)$ is the best response to agent $i$ in sub-case 2.2. \\

\textbf{Case 3:} \label{para_case_03_new} $\theta_i = A$, which signifies that agent $i$ is an attester.

\textbf{Hypothesis:} Any agent $i$ still plays a pure strategy of either $a_i = C$ or $a_i = D$. However, all agents $-i$ play mixed strategies.

\textbf{Analysis:} This case examines the impact of the "inactivity leak" from some agents $-i$ playing $s_{-i}(D)$ on agent $i$ with $\theta_i = A$. Here, one of the agents $-i$ is the block proposer, and the others are attesters. We will focus on the strategies of the block proposer among agents $-i$. \\

\textbf{Sub-case 3.1}: The condition $x^{BC} = 1$ holds, indicating that a unique validator among agents $-i$ executes $s_{-i}(C)$ with $\theta_{-i} = B$. In this sub-case, we continue to conduct our analysis based on other agents $-i$'s strategies in two main scenarios. \\

\textbf{Scenario 3.1.1:} More than or equal to $\frac{1}{3}$ agents $-i$ choose $s_{-i}(D)$, and others play $s_{-i}(C)$ with $\theta_{-i} = A$. Similar to our analysis in sub-cases 2.1 and 2.2, we need to consider the impact of the "inactivity leak" to calculate agent $i$'s utility. 

If the "inactivity leak" is triggered, we can employ the following equations to compute agent $i$'s utility by assuming that $\gamma \geq \frac{1}{3}$ and agent $i$ play either $C$ or $D$, we can utilize the following equation to compute the utility of agent $i$:

\begin{equation} \label{expected_utilities_subcase_3_1_1_01}
\begin{split}
&EU_i(C, s_{-i}, A, A \mid \gamma \geq \frac{1}{3},~x^{BC}=1) \\
&= p(A, A \mid \gamma \geq \frac{1}{3},~x^{BC}=1)~\cdot \\
& \quad (u_i(C, a_{-i}, A, A \mid \gamma \geq \frac{1}{3},~x^{BC}=1) - C_i^{EP}(y)) \\
&=  p(A, A \mid \gamma \geq \frac{1}{3}, x^{BC}=1) \cdot (-~C_i^{EP}(y))
\end{split}
\end{equation}

\begin{equation} \label{expected_utilities_subcase_3_1_1_02}
\begin{split}
& EU_i(D, s_{-i}, A, A \mid \gamma \geq \frac{1}{3}, x^{BC}=1) \\ 
&= p(A, A \mid \gamma \geq \frac{1}{3}, x^{BC}=1) \cdot \\
& \quad (u_i(D, a_{-i}, A, A \mid \gamma \geq \frac{1}{3}, x^{BC}=1) - C_i^{EP}(y)) \\
&=  p(A, A \mid \gamma \geq \frac{1}{3}, x^{BC}=1) \cdot (-P_i^{IL}-~C_i^{EP}(y))
\end{split}
\end{equation}

In equation \ref{expected_utilities_subcase_3_1_1_01}, agent $i$ cannot receive rewards, as Eth2 does not allocate rewards to attesters during the "inactivity leak" period. Furthermore, according to Equation \ref{expected_utilities_subcase_3_1_1_02}, agent $i$ incurs inactivity penalties ($P_i^{IL}$) due to becoming offline. 

In contrast, if the "inactivity leak" is not triggered, we can employ the following equations to calculate agent $i$'s utility:

\begin{equation} \label{expected_utilities_subcase_3_1_1_03}
\begin{split}
& EU_i(C, s_{-i}, A, A \mid \gamma \geq \frac{1}{3}, x^{BC}=1) \\
&= p(A, A \mid \gamma \geq \frac{1}{3}, x^{BC}=1) ~\cdot \\
& \quad (u_i(C, a_{-i}, A, A \mid \gamma \geq \frac{1}{3}, x^{BC}=1) - C_i^{EP}(y)) \\
&=  p(A, A \mid \gamma \geq \frac{1}{3}, x^{BC}=1) \cdot (R_i^A -~C_i^{EP}(y))
\end{split}
\end{equation}

\begin{equation} \label{expected_utilities_subcase_3_1_1_04}
\begin{split}
& EU_i(D, s_{-i}, A, A \mid \gamma \geq \frac{1}{3}, x^{BC}=1) \\
&= p(A, A \mid \gamma \geq \frac{1}{3}, x^{BC}=1)~\cdot \\
& \quad (u_i(D, a_{-i}, A, A \mid \gamma \geq \frac{1}{3}, x^{BC}=1) - C_i^{EP}(y)) \\
&=  p(A, A \mid \gamma \geq \frac{1}{3}, x^{BC}=1) \cdot (-P_i^{A} - C_i^{EP}(y))
\end{split}
\end{equation}

In Equation \ref{expected_utilities_subcase_3_1_1_03}, since the "inactivity leak" is not triggered, agent $i$ still remains eligible to receive rewards from Eth2 for completing their tasks as an attester. However, as outlined in Equation \ref{expected_utilities_subcase_3_1_1_04}, agent $i$ will incur penalties for turning off its device and failing to execute its task.\\

\textbf{Scenario 3.1.2:} Less than $\frac{1}{3}$ agents $-i$ choose $s_{-i}(D)$ and others play $s_{-i}(C)$ with $\theta_{-i} = A$. In this scenario, when agent $i$ plays $C$, the utility is computed by Equation \ref{expected_utilities_subcase_3_1_1_05}. Also, we can employ Equation \ref{expected_utilities_subcase_3_1_1_06} to calculate agent $i$'s utility when it plays $D$ in this sub-case. \\

\begin{equation} \label{expected_utilities_subcase_3_1_1_05}
\begin{split}
& EU_i(C, s_{-i}, A, A \mid \gamma < \frac{1}{3}, x^{BC}=1) \\
&= p(A, A \mid \gamma < \frac{1}{3}, x^{BC}=1) ~\cdot \\
& \quad (u_i(C, a_{-i}, A, A \mid \gamma < \frac{1}{3}, x^{BC}=1) - C_i^{EP}(y)) \\
&=  p(A, A \mid \gamma < \frac{1}{3}, x^{BC}=1) \cdot (R_i^A -~C_i^{EP}(y))
\end{split}
\end{equation}

\begin{equation} \label{expected_utilities_subcase_3_1_1_06}
\begin{split}
& EU_i(D, s_{-i}, A, A \mid \gamma < \frac{1}{3}, x^{BC}=1) \\
&= p(A, A \mid \gamma < \frac{1}{3}, x^{BC}=1)~\cdot \\
& \quad (u_i(D, a_{-i}, A, A \mid \gamma < \frac{1}{3}, x^{BC}=1) - C_i^{EP}(y)) \\
&=  p(A, A \mid \gamma < \frac{1}{3}, x^{BC}=1) \cdot (-P_i^{A} - C_i^{EP}(y))
\end{split}
\end{equation}

According to Equations \ref{expected_utilities_subcase_3_1_1_05} and \ref{expected_utilities_subcase_3_1_1_06}, agent $i$ acquires rewards for remaining online and incurs penalties for turning its device offline.

Thus far, based on the analysis from Equations \ref{expected_utilities_subcase_3_1_1_01} to \ref{expected_utilities_subcase_3_1_1_06}, we can obviously conclude that playing $s_i^*(C)$ is the best response to agent $i$ in sub-case 3.1. \\

\textbf{Sub-case 3.2}: The condition $x^{BD} = 1$ holds, indicating that a unique validator among agents $-i$ executes $s_{-i}(D)$ with $\theta_{-i} = B$. In this sub-case, we mainly discuss the varies of agent $i$'s utility if the block proposer is played by one of the agents $-i$. Similarly, we must consider whether the "inactivity leak" is triggered in the following two scenarios: \\

\textbf{Scenario 3.2.1:} More than or equal to $\frac{1}{3}$ agents $-i$ choose $s_{-i}(D)$, and others play $s_{-i}(C)$ with $\theta_{-i} = A$.

In this scenario, if the "inactivity leak" is triggered, we can employ the following equations to calculate agent $i$'s utility:

\begin{equation} \label{expected_utilities_subcase_3_2_1_01}
\begin{split}
& EU_i(C, s_{-i}, A, A \mid \gamma \geq \frac{1}{3}, x^{BD}=1) \\
&= p(A, A \mid \gamma \geq \frac{1}{3}, x^{BD}=1)~\cdot \\
& \quad (u_i(C, a_{-i}, A, A \mid \gamma \geq \frac{1}{3}, - C_i^{EP}(y)) \\
&= p(A, A \mid \gamma \geq \frac{1}{3}, x^{BD}=1) \cdot (- C_i^{EP}(y))
\end{split}
\end{equation}

\begin{equation} \label{expected_utilities_subcase_3_2_1_02}
\begin{split}
& EU_i(D, s_{-i}, A, A \mid \gamma \geq \frac{1}{3}, x^{BD}=1) \\
&= p(A, A \mid \gamma \geq \frac{1}{3}, x^{BD}=1)~\cdot \\ & \quad (u_i(D, a_{-i}, A, A \mid \gamma \geq \frac{1}{3}, x^{BD}=1) - C_i^{EP}(y)) \\
&=  p(A, A \mid \gamma \geq \frac{1}{3}, x^{BD}=1) \cdot (-P_i^{IL}-~C_i^{EP}(y))
\end{split}
\end{equation}

In Equation \ref{expected_utilities_subcase_3_2_1_01}, although agent $i$ completes its task as an attester, it still cannot receive a reward since the "inactivity leak" is triggered. Additionally, agent $i$ incurs inactivity penalties as outlined in Equation \ref{expected_utilities_subcase_3_2_1_02}. On the contrary, if the "inactivity leak" is not triggered, we can utilize the following equations to compute agent $i$'s utility:

\begin{equation} \label{expected_utilities_subcase_3_2_1_03}
\begin{split}
& EU_i(C, s_{-i}, A, A \mid \gamma \geq \frac{1}{3}, x^{BD}=1) \\
&= p(A, A \mid \gamma \geq \frac{1}{3}, x^{BD}=1)~\cdot \\ & \quad (u_i(C, a_{-i}, A, A \mid \gamma \geq \frac{1}{3}, x^{BD}=1) - C_i^{EP}(y)) \\
&=  p(A, A \mid \gamma \geq \frac{1}{3}, x^{BD}=1) \cdot (R_i^A - C_i^{EP}(y))
\end{split}
\end{equation}

\begin{equation} \label{expected_utilities_subcase_3_2_1_04}
\begin{split}
& EU_i(D, s_{-i}, A, A \mid \gamma \geq \frac{1}{3}, x^{BD}=1) \\
&= p(A, A \mid \gamma \geq \frac{1}{3}, x^{BD}=1)~\cdot \\
& \quad (u_i(D, a_{-i}, A, A \mid \gamma \geq \frac{1}{3}, x^{BD}=1) - C_i^{EP}(y)) \\
&=  p(A, A \mid \gamma \geq \frac{1}{3}, x^{BD}=1) \cdot (-P_i^{A} - C_i^{EP}(y))
\end{split}
\end{equation}

Although the block proposer does not propose a new block, agent $i$, as an attester, still receives a reward or penalty depending on whether it completes its task. Considering Equations \ref{expected_utilities_subcase_3_2_1_01}, \ref{expected_utilities_subcase_3_2_1_02}, \ref{expected_utilities_subcase_3_2_1_03}, and \ref{expected_utilities_subcase_3_2_1_04}, we can conclude that playing $s_i^*(C)$ remains the best response for agent $i$ in scenario 3.2.1. \\

\textbf{Scenario 3.2.2:} Less than $\frac{1}{3}$ agents $-i$ choose $s_{-i}(D)$ and others play $s_{-i}(C)$ with $\theta_{-i} = A$. In this scenario, since the "inactivity leak" will not be triggered, we can employ equations \ref{expected_utilities_subcase_3_2_1_05} and \ref{expected_utilities_subcase_3_2_1_06} to calculate agent $i$'s utilities when it plays $C$ and $D$, respectively. \\

\begin{equation} \label{expected_utilities_subcase_3_2_1_05}
\begin{split}
& EU_i(C, s_{-i}, A, A \mid \gamma < \frac{1}{3}, x^{BD}=1) \\
&= p(A, A \mid \gamma < \frac{1}{3}, x^{BD}=1)~\cdot \\ & \quad (u_i(C, a_{-i}, A, A \mid \gamma < \frac{1}{3}, x^{BD}=1) \cdot (R_i^A - C_i^{EP}(y)) \\
&=  p(A, A \mid \gamma < \frac{1}{3}, x^{BD}=1) \cdot (R_i^A - C_i^{EP}(y))
\end{split}
\end{equation}

\begin{equation} \label{expected_utilities_subcase_3_2_1_06}
\begin{split}
& EU_i(D, s_{-i}, A, A \mid \gamma < \frac{1}{3}, x^{BD}=1) \\
&= p(A, A \mid \gamma < \frac{1}{3}, x^{BD}=1)~\cdot \\
& \quad (u_i(D, a_{-i}, A, A \mid \gamma < \frac{1}{3}, x^{BD}=1) - C_i^{EP}(y)) \\
&=  p(A, A \mid \gamma < \frac{1}{3}, x^{BD}=1) \cdot (-P_i^{A} - C_i^{EP}(y))
\end{split}
\end{equation}

This scenario is a normal case where agent $i$ will receive rewards and incur penalties whether they complete or fail to complete the task of being an attester, as shown in Equations \ref{expected_utilities_subcase_3_2_1_05} and \ref{expected_utilities_subcase_3_2_1_06}, respectively. \\ 

Consequently, based on Equations \ref{expected_utilities_subcase_3_2_1_01} to  \ref{expected_utilities_subcase_3_2_1_06}, playing $s_i^*(C)$ is still the best response to agent $i$ in sub-case 3.2. \\

\textbf{Case 4:} Agent $i$ plays a mixed strategy.

\textbf{Hypothesis:} Agent $i$ and all agents $-i$ play mixed strategies.

\textbf{Analysis:} This case is used to analyze the scenario where, instead of executing a pure strategy as demonstrated in Cases 1 to 3, agent $i$ might execute mixed strategies depending on its type. For example, regardless of whether $\theta_i$ is $B$ or $A$, it will execute $C$ or $D$ based on different probabilities. In this situation, we can model agent $i$'s expected utility as:

\begin{equation} \label{eq_eu5_mixed}
    EU_i(s_{i,\theta}, s_{-i,\theta}) = \sum_{\theta \in \Theta} p(\theta) \cdot \left[ p_C \cdot u_i(C, a_{-i}, \theta) + p_D \cdot u_i(D, a_{-i}, \theta) \right]
\end{equation}

where $p_C$ and $p_D$ are the probabilities that agent $i$ plays $C$ or $D$, respectively. 

According to our proof from Cases 1 to 3, we have already proven that when agent $i$ executes $D$ in various sub-cases, it might incur penalties from Eth2 or at least result in not receiving rewards from Eth2. As a result, the value yielded by Equation \ref{eq_eu5_mixed} when agent $i$ plays $D$ is always less than when agent $i$ plays $C$, proving that playing $C$ is the dominant strategy for agent $i$, regardless of the strategies all agents $-i$ take.

In summary, across Cases 1 to 4, it is evident that the best response for any agent $i$, denoted as $BR_i$, is to adopt strategy $s_i^*(C \mid \theta_i)$, irrespective of the strategies of agents $-i$. Given the symmetry of the utility function across agents, this implies that $BR_{-i}$ similarly defaults to $s_{-i}(C \mid \theta_{-i})$. Thus, the Bayesian Nash Equilibrium (BNE) for this game, under Eth2's original reward scheme, is $s^*(C)$, where all agents select strategy $C$ regardless of their types, with no agent able to unilaterally deviate from improving their expected utility. Moreover, the analysis from all cases confirms that for every $\theta_i \in \Theta_i$, $u_i(C, s_{-i}, \theta_i)$ consistently exceeds $u_i(D, s_{-i}, \theta_i)$, where $s_{-i}$ represents the collective strategy profile of all other agents. Consequently, $C$ emerges as the ex ante dominant strategy for any agent $i$, independent of the strategies of agents $-i$.
\end{proof}

\subsection{Implications of Incentive Analysis for Eth2’s Multiple Block Proposers and Attesters
} \label{sect_incentive_analysis_discussion} 
Based on theorem \ref{theorem_bne_cooperate_eth}, we know strategy profile $s_i^*(C)$ is an \emph{ex ante dominant strategy}. Then we could have the next corollary: 

\begin{corollary}
\label{corollary_incentive_compatibility_eth}
In the Eth2 Bayesian game described in Theorem \ref{theorem_bne_cooperate_eth}, the strategy profile $s_i^*(C)$ constitutes an ex ante dominant strategy equilibrium. Consequently, the Eth2 protocol ensures incentive compatibility for interactions between any agent $i$ and the collective of other agents $-i$, in that playing $C$ is the best response for each agent when they adhere to the protocol, irrespective of their own type or the strategies of other agents.
\end{corollary}

Note that we utilize the previously proven Theorem \ref{theorem_bne_cooperate_eth} to assert that playing $C$ is an ex ante dominant strategy. Consequently, we conclude that the Eth2 Bayesian game, defined in this paper, is incentive-compatible, signifying that no agent has the incentive to deviate from playing $C$. \\

Overall, this paper presents three main contributions: \textbf{First}, Theorem \ref{theorem_bne_cooperate_eth} demonstrates that when agent $i$ plays $s_i^*(C)$, this strategy constitutes a BNE. \textbf{Second}, insights from Theorem \ref{theorem_bne_cooperate_eth} confirm that choosing $C$ is an ex ante dominant strategy for any agent $i$. \textbf{Finally}, through Corollary \ref{corollary_incentive_compatibility_eth}, we show that playing $C$ is the best response for agent $i$, regardless of the agent's type or the strategies of other agents, which further indicates that Eth2's protocol holds incentive compatibility.

\subsection{Open Questions to Our Model} \label{sect_incentive_analysis_openquestions} 
In this section, we explore several open questions that pertain to our game-theoretic analysis of Eth2.0 validator strategies. These open questions relate to specific aspects of our game model and external factors that might affect the generalizability and applicability of our findings. Acknowledging and understanding these open questions is crucial for accurately interpreting our results and may provide valuable insights for future research in algorithmic game theory and blockchain economics. 

\subsubsection{Maximal extractable value (MEV)}
In Eth2, maximal \linebreak extractable value (MEV) \cite{Ethereum_mev} allows validators to extract additional value by manipulating transaction orders within their proposed blocks, a practice that could destabilize the Eth2 consensus protocol if not managed correctly \cite{daian2020flash, Ethereum_mev}. To address the negative impact from MEV, the \emph{Proposer-builder separation} (PBS) \cite{Ethereum_PBS} has been suggested, which separates the roles of block proposers and builders to mitigate aggressive competition for MEV \cite{Ethereum_PBS, heimbach2023ethereum, wu2023strategic}. Although PBS has not yet been officially integrated into Eth2, tools like MEV-Boost \cite{Flashbot_mev_boost}, developed by Flashbots, have already demonstrated \linebreak their practical application. In fact, MEV-Boost is precisely a solution that Flashbots aims to address Eth2 MEV issues before the official adoption of PBS \cite{Flashbot_mev_boost, Ethereum_mev, wu2023strategic}. A comprehensive examination of MEV and PBS with Eth2 could provide valuable insights into their implications for validator strategies. However, such an analysis is beyond the scope of this paper and is recommended for future research.

\subsubsection{Delegated Proof of Stake (DPoS)}
Given that the requirement to deposit 32 ETH constitutes a significant barrier for general users wishing to become validators in Eth2, the DPoS has been proposed. DPoS facilitates participation by following users to stake varying amounts of ETH, constrained by the minimum and maximum limits set by DPoS platforms, such as  Coinbase \cite{DPoSCoinbase}, Kraken \cite{DPoSKraken}, or Crypto.com \cite{DPoSCryptocom}. Consequently, this enables broader participation in the validator process within Eth2 \cite{zhang2023rationally}. Additionally, DPoS provides a mechanism whereby users may transfer some of the financial burdens and potential penalties associated with the role of validator to the DPoS providers \cite{Ethereum_DPoS}. This transferability is likely to influence the strategic decisions of validators. Integrating DPoS as several actionable strategies within our model represents a promising direction for future expansion. 

\subsubsection{External Factors}
In this paper, we posit that validators choose their strategies based on the principle of individual rationality, striving to maximize their utilities \cite{von1947theory}. Based on this assumption, we hypothesize that should the incentives provided in Eth2 prove insufficient, validators are likely to disengage from the platform \cite{matsunaga2022incentivization}. Consequently, it is imperative for future research to investigate the impacts of external factors such as economic fluctuations, competition from other financial assets, and government regulatory policies on the strategic decisions of validators. Such investigations will benefit our model and ensure its alignment with real-world applications, thereby yielding more detailed insights into the strategic behaviors of Eth2 validators. 
\section{Related work}
\label{sect_related_work}

\textbf{Security Defenses of Ethereum 2.0}. Ethereum has already transitioned its consensus mechanism from \emph{PoW} to \emph{PoS}, defined its incentive mechanism, and stated how to discourage attacks on \linebreak Ethereum \cite{buterin2018discouragement, buterin2020incentives}. However, Ethereum is a complex platform, and its smart contract language is designed to support decentralized applications. Because of this complexity, many potential security vulnerabilities in Ethereum need to be further studied and evaluated. Researchers have discovered several vulnerabilities, described what attacks target these weaknesses of Ethereum, and then discussed how to defend against these threats. Some threats have already been resolved, yet there are still threats that need practical solutions. The \emph{Nothing-at-Stake} and \emph{Long-range} attacks are two main threats to \emph{PoS-based} blockchains \cite{chen2020survey}.

\textbf{\emph{Nothing-at-Stake}} attacks represent that because \emph{PoS} does not need to calculate the cost for mining as \emph{PoW} does, verifiers can create multiple blocks at multiple forks or participate in the verification without cost to improve their probability of obtaining benefits. This attack will lead to more and more forks in the \emph{PoS} blockchain and ultimately lead to the inability of all nodes to reach a consensus. \textbf{\emph{Long-range}} attacks refer to an attacker's attempt to establish a fork from a specific block and make it the longest chain to replace the current legal longest in the PoS blockchain. One possible way is for an attacker to obtain the private keys that were originally used to create blocks. These private keys are worthless for some verifiers who want to retrieve staking tokens. However, attackers can manipulate these private keys to create another longest chain from any existing block. 

To mitigate threats from "Nothing-at-Stake" and "Long-range" attacks, Ethereum adopts two approaches. Firstly, to prevent \linebreak "Nothing-at-Stake" attacks, Ethereum has increased the threshold for becoming a "validator" (currently 32 ETH) and imposes high penalties on validators who execute such attacks. As Saleh (2021) mentions, this mechanism effectively prevents "Nothing-at-Stake" attacks by raising the entry threshold and reducing rewards \cite{saleh2021blockchain}. Since participants hold tokens, they desire transactions to be confirmed quickly to secure long-term rewards. Secondly, Ethereum uses checkpoint and head block mechanisms to periodically determine the first block of each epoch, thereby avoiding "Long-range" attacks \cite{vitalikannotatedspec}. In fact, the implementation of these checkpoint and head block mechanisms critically depends on validators. Therefore, to encourage validators to complete tasks assigned in the consensus mechanism correctly and to make correct votes, reducing malicious behavior, Eth2 has designed a set of rewards and penalties, which constitute the incentive mechanism of ETH2. \\

\noindent \textbf{Analysis of Blockchain Participant Behavior.} However, the \emph{Verifier’s Dilemma} poses a significant challenge to the incentive mechanisms of PoW and PoS blockchains, as rational verifiers may deviate from consensus protocols to prioritize tasks offering the most significant benefits, influenced by the varying rewards for each task \cite{chen2020survey}. As noted by Aldweesh et al. (2018), miners may selectively prioritize specific tasks when income is not proportional to costs, potentially causing negative impacts on the blockchain's uninterrupted and dependable operation \cite{aldweesh2018performance}. Additionally, Alharby et al. (2020) discovered that rational miners are more inclined to mine new blocks rather than verify blocks mined by others based on analyses of existing smart contracts \cite{alharby2020data}. In this paper, we explore the following research question: \textbf{Is it better for validators in Eth2 to stay online continuously or to strategically go offline during certain slots, considering the potential benefits and risks associated with each strategy?} To address this question, we utilize a Bayesian game to determine the best strategy for validators to maximize their rewards and maintain Eth2's integrity. \\

\noindent \textbf{Game Theory Analysis in PoS Blockchains.} Researchers have devoted considerable effort to ensuring incentive compatibility in PoS-based blockchains using game theory techniques, which is pertinent to the research questions and methodologies of our study. For example, Saleh (2021) proposed, from the economics and game theory perspective, that a suitable \emph{PoS} incentive mechanism should meet two criteria: setting a sufficiently high threshold of becoming validators and effectively reducing block reward to promote consensus. Saleh used an extensive form game to analyze whether Eth2's reward scheme could resist the Nothing-at-stake attack and discussed potential issues, such as wealth concentration, double spending, and staking pooling based on their findings \cite{saleh2021blockchain}. However, Saleh did not consider validator types in a Bayesian game, which is the game model that the authors plan to study in their research. On the other hand, Motepalli and Jacobsen (2021) adapted the evolutionary game theory to analyze how participants' behavior can evolve with the reward mechanism. They concluded that penalties play a central role in maintaining the integrity of \linebreak blockchains. Their key result is that punishment is essential in building mechanisms for PoS blockchains, but their paper does not focus on analyzing Eth2's mechanism. It does not consider behaviors from different validator types \cite{motepalli2021reward}. Additionally, Salau et al. (2022) utilized an infinitely repeated game model to examine the competitive interactions among Eth2’s validators and ensure they act honestly to avoid malicious activities by encouraging validators to report deviation actions and setting the discount factor as close to one as possible \cite{9921952}. 

However, setting the discount factor close to one is not always reasonable in actual PoS blockchains, as validators might favor short-term gains over long-term profits due to reasons like liquidity needs or risk preferences. In contrast, we believe that employing a Bayesian game will allow us to ascertain that cooperation remains a dominant strategy in a one-shot game that models the Eth2, an analysis more in line with the real-world PoS blockchains. Further, Bhudia et al. (2023) utilized a sub-game to model paying ransom to avoid slashing in Eth2, which is the best response for validators who are threatened by attackers \cite{bhudia2023game}. Nonetheless, they assume the attack will be successful only if the attackers build a smart contract to interact with the victims, which might not be entirely reasonable in the context of the real-world Eth2. Also, in \cite{bhudia2023game}, a Bayesian game, which could potentially provide a more comprehensive analysis compared to a sub-game, given that the scope of a sub-game is limited to proposed sub-cases, was not employed. Similarly, Chaidos et al. (2023) indicate that when blockchain systems allow users to "retract" their commitment - that is, users initially commit to performing the tasks of validators but later can choose whether or not to fulfill the assigned tasks - the system can achieve a state of equilibrium. This study underscores the importance of designing suitable incentive mechanisms to encourage users to actively participate and fulfill the responsibilities of validators, even in the presence of potential free-riding behavior \cite{chaidos2023blockchain}. However, the classical strategic game has certain limitations in expressing the impact of user heterogeneity on strategy choices. When analyzing equilibrium, considering the impact of different user types on strategy choices is crucial. Bayesian games, by introducing different types of users and the probability distribution of their strategy choices, can capture this heterogeneity in a more nuanced way. Moreover, compared to the classical strategic game, which assumes that each player has complete information, Bayesian games consider a more realistic scenario where players often have incomplete information in real-world situations. \\

\noindent \textbf{Incorporating Bayesian Approaches in Blockchain Mechanism Design.} Given that Bayesian models provide a robust framework for handling incomplete information, such models are particularly suitable for analyzing validators' strategies in blockchain systems or developing blockchain-based solutions in various domains \cite{liu2019survey, shi2022integration}. Liu et al. (2019) discussed how a Bayesian game could be used to model not only the risk level of blockchain \linebreak providers but also the incomplete information of miners and service providers in edge computing \cite{liu2019survey}. Yan et al. (2020) \cite{yan2020dynamic} and Chen et al. (2022) \cite{chen2022bayesian} proposed Bayesian game approaches to develop a transaction fee allocation mechanism for blockchains. Additionally, Xia et al. applied a Bayesian Game-based solution to determine the electricity pricing in a vehicle-to-vehicle electricity trading scheme \cite{xia2020bayesian}. Zhang et al. (2023) utilized a Bayesian network model to develop a \emph{Validator Selection Game}, helping clients decide how to choose their DPoS validators \cite{zhang2023rationally}. Liao et al. (2023) introduced a novel reward scheme to address the free-rider issue in Algorand's incentive mechanism, modeled as a Bayesian game, ensuring that participation in the protocol achieves a Bayesian Nash equilibrium, even with a malicious adversary present \cite{liao2023irs}. These works demonstrate how Bayesian approaches can be effectively applied to blockchain systems to model the incomplete information among agents, similar to our approach in this paper, where we model interactions between validators in Eth2. Consequently, we employ a Bayesian model to investigate whether Eth2 possesses incentive compatibility by verifying the presence of a BNE and an ex ante dominant strategy when validators interact. Our findings suggest that for validators, deviating from the Eth2 protocol to take their devices offline to avoid specific tasks and save costs is not the best response.

\section{Conclusion and Future Works}
\label{sect_conclusion}
In this paper, we achieved three main contributions: First, we comprehensively summarized essential knowledge regarding the rewards, penalties, and cost models of Eth2 from various sources. This serves as a foundation for our readers to understand the context of our study. Second, we utilized a Bayesian game framework to model the strategies of validators to explore the BNE and dominant strategies in a single slot of Eth2. Our analysis demonstrates that cooperation (i.e., keeping their devices online) can achieve a BNE and an ex ante dominant strategy in Eth2. Lastly, by establishing an ex ante dominant strategy in Eth2, we derived a corollary asserting that Eth2’s incentive mechanism is incentive-compatible. Our analysis, grounded in a rigorous game model and detailed calculations, provides an objective and robust substantiation of our claim while laying a sturdy groundwork for future examinations of Eth2's incentive mechanism and offering valuable insights for individuals contemplating participation as validators in Eth2. However, aside from the open questions mentioned in Section \ref{sect_incentive_analysis_openquestions}, we currently face two limitations: First, the exclusion of the validators' roles in the sync committee; second, the need to extend our model to include a more detailed analysis of validators' strategies over multiple epochs, since past rewards or penalties might influence future validator strategies, rather than treating each epoch as isolated. These areas present opportunities for future exploration. Despite these open questions and limitations, we firmly believe this research lays a solid foundation for further exploring the impact of validators’ strategies in Eth2.


\begin{acks}
This study was funded by Ripple Labs and the Natural Sciences and Engineering Research Council of Canada (NSERC). We also thank Dr. Seyed Majid Zahedi, Dr. Kate Larson, and Dr. Mahesh Tripunitara \wg{at} the University of Waterloo for providing valuable insights on earlier versions of this work. Additionally, we are grateful to the anonymous reviewers for their valuable suggestions and feedback.
\end{acks}

\bibliographystyle{ACM-Reference-Format}
\bibliography{sample-base}


\end{document}